\documentclass[12pt]{iopart}
\pdfoutput=1

\expandafter\let\csname equation*\endcsname\relax
  \expandafter\let\csname endequation*\endcsname\relax
\usepackage{amsmath}
\usepackage{amsthm}
\usepackage{iopams,subfigure}
\usepackage{graphicx,color}

\newcommand{\Rst}{{\mathbb{R}}}
\newcommand{\bC}{{\mathbb{C}}}
\newcommand{\Rdst}{{\mathbb{R}^d}}
\newcommand{\norm}[1]{\lVert#1\rVert}
\newcommand{\trace}{\mathrm{trace}}

\newcommand{\abs}[1]{\ensuremath{\left| #1 \right| }}
\newcommand{\newd}{m}
\newtheorem{lemma}{Lemma}[section]
\newtheorem{theorem}[lemma]{Theorem}
\newtheorem{coro}[lemma]{Corollary}

\newtheorem{remark}[lemma]{Remark}
\newtheorem{definition}[lemma]{Definition}

\begin{document}
\title[The WH ensemble: hyperuniformity and higher Landau levels]{The Weyl-Heisenberg ensemble: hyperuniformity and
higher Landau levels}

\author{L. D. Abreu$^1$, J. M. Pereira$^2$, J. L. Romero$^1$ and S. Torquato$^3$}
\address{$^1$ Austrian Academy of Sciences, Acoustics Research Institute, Wohllebengasse 12-14, Vienna A-1040, Austria.}

\address{$^2$ Program in Applied and Computational Mathematics, Princeton University, New Jersey 08544, USA}

\address{$^3$ Department of Chemistry, Department of Physics, Princeton Institute for the Science
and Technology of Materials, and Program in Applied and Computational
Mathematics, Princeton University, New Jersey 08544, USA}

\begin{abstract}
Weyl-Heisenberg ensembles are a class of determinantal point processes
associated with the Schr\"{o}dinger representation of the Heisenberg
group. Hyperuniformity characterizes a state of matter for which (scaled)
density fluctuations diminish towards zero at the largest length scales. We
will prove that Weyl-Heisenberg ensembles are hyperuniform. Weyl-Heisenberg ensembles include as a
special case a multi-layer extension
of the Ginibre ensemble modeling the distribution of electrons in higher
Landau levels, which has recently been object of study in the realm of the
Ginibre-type ensembles associated with polyanalytic functions. In addition, the family of Weyl-Heisenberg ensembles includes new structurally anisotropic processes, where point-statistics depend on
 the different spatial directions, and thus provide a first means to study directional hyperuniformity.
\end{abstract}


\bigskip

\bigskip

\hspace{1.55cm} \parbox{12cm}{{\it Keywords\/}: Weyl-Heisenberg ensemble, hyperuniformity, higher Landau levels,
disorder, fluctuations.}

\maketitle

\section{Introduction}
The characterization of density fluctuations in many-body systems is a
problem of great interest in the physical, mathematical and biological
sciences. A \emph{hyperuniform} many-particle system is one in which
density fluctuations are anomalously suppressed at long-wavelengths,
compared to those occurring in the Poisson point process and typical
correlated disordered point processes. The hyperuniformity concept provides
a new way to classify crystals, certain quasiperiodic systems, and special
disordered systems \cite{TorStil,PHRE2009}. A theory for
understanding hyperuniformity in terms of the number variance of point
processes has been developed in \cite{TorStil} and \cite{PRL2011}. The
theory characterizes hyperuniform point process in $d$-dimensions with the
property that the variance in the number of points in an observation window
of radius $R$ grows at a rate slower than $R^{d}$, which is the growth
rate for a Poisson point process.

It is known that some determinantal point processes are disordered
and hyperuniform \cite{TorScaZac,CosLeb,SM2009,PNAS2013,GhLe16}. \emph{Weyl-Heisenberg Ensembles}
are a very general class of determinantal point processes
on $\mathbb{R}^{d}$, with $d=2 \newd$ an even number. They are
defined in terms of the Schr\"{o}dinger
representation of the Heisenberg group acting on a vector
$g\in L^{2}(\mathbb{R}^{\newd})$. For choices of $g$ in the
Hermite function basis of
$L^{2}(\mathbb{R}^{\newd})$, they reduce to extensions of the two-dimensional
one-component plasma and the Ginibre ensemble to higher Landau levels. In
this paper, we will show that Weyl-Heisenberg ensembles are hyperuniform.
Actually, a bit more is true: the variance of the number of points in an
observation window of radius $R$ grows at a rate proportional to $R^{d-1}$.
This was already known to happen in the two-dimensional one-component
plasma. Our results about the Weyl-Heisenberg ensembles show that the same
happens with the distribution of electrons when higher Landau levels are
formed under strong magnetic fields.
Higher Landau levels lead to the macroscopic
effect known as the Quantum-Hall effect \cite{GirvJach,QHEGRAPH}.

There are several ways of randomly distributing points on the Euclidean space. Under mild assumptions, a point process is completely specified by the countably
infinite set of generic $k$-particle probability density functions, denoted
by $\rho _{k}(\mathbf{r}_{1},...,\mathbf{r}_{k})$. These are proportional to the
probability density of finding $k$ particles in volume elements around the
given positions $(\mathbf{r}_{1},...,\mathbf{r}_{k})$, irrespective of the
remaining particles. More precisely, if $\mathcal{X}$ is a simple process defined on the Euclidean space - so that points do not have multiplicities, then $\rho_k$ is characterized by the properties that (i) $\rho _{k}(\mathbf{r}_{1},...,\mathbf{r}_{k})=0$, whenever
the positions of two of the
points
$\mathbf{r}_{j}$ are equal, and (ii) for every family of disjoint measurable sets $D_1, \ldots D_k$,
\begin{align*}
\mathbb{E} \left[ \prod_{j=1}^k \mathcal{X}(D_j) \right]
= \int_{\prod_j D_j} \rho_{k}(\mathbf{r}_{1},...,\mathbf{r}_{k})
d\mathbf{r}_{1} \ldots d\mathbf{r}_{k},
\end{align*}
where $\mathcal{X}(D)$ denotes the number of points to be found in $D$.
For instance, Poisson processes distribute points randomly but in a
completely uncorrelated way. Indeed, in this case, the $k$-particle
probability density reduces to constants easy to compute. Since the
single-particle density function can be obtained from the thermodynamical
limit $\rho $, where $V$ is the strength of the magnetic field and grows
together with the number of points $N$:
\begin{equation*}
\rho (\mathbf{r}_{1})=\lim_{N,V\rightarrow \infty }\frac{N}{V}=\rho ,
\end{equation*}
the $k$-particle probability densities for Poisson processes are given by
\begin{equation*}
\rho _{k}(\mathbf{r}_{1},...,\mathbf{r}_{k})=\rho ^{k}.
\end{equation*}
However, in several many-body systems and other physical models, one has to
take into account particle-particle interactions, requiring more
sophisticated probabilistic models. In studies of the statistical mechanics
of point-like particles one is usually interested in a handful of quantities
such as $k$-particle correlations. It is then of paramount importance to
study point processes for which the properties of such correlations have
convenient analytic descriptions as, for instance, in the so-called ghost
random sequential addition processes \cite{TorStil2006}. But such exactly
solvable models are not so common, leading to widespread use of Poisson
processes instead of more sophisticated probability models, because simple
analytic expressions for the $k$-particle correlations are available. Facing
such a gap between the physical model and its mathematical description, one
may be led to think that the possibility of using probabilistic models
describing interacting particles, with $k$-particle correlations written in
analytic form, is a hopeless mathematical chimera. However,
using determinantal point processes, it is possible to construct such probabilistic models.

Determinantal Point Processes (DPP's) are defined in terms of a kernel where the negative correlation between points is
built in.
Because of the repulsion inherent of the model, they are convenient to describe physical systems with charged-liked
particles, where the confinement to a
bounded region is controlled by a weight function involving the external
field. Unlike other non-trivial statistical models, the $k$-particle
correlations of DPP's admit an analytic expression as determinants whose
entries are defined using the correlation kernel. Moreover, DPP's enjoy
a remarkable property which allows one to derive the macroscopic laws of
physical systems constituted by interacting particles which display chaotic
random behavior at small scales: for very large systems of points
confined to bounded regions, the distribution patterns begin to look less
chaotic and start to organize themselves in an almost uniform way. In
several cases, in the proper scaled thermodynamical limit, the distributions
are either uniform or given by analytic expressions. This phenomenon is
related to a property of physical and mathematical systems known as
\emph{universality} \cite{Deift}.

The present paper explores a link between the theory of DPP's and the
Schr\"{o}dinger representation of the Heisenberg
group. This link is
important because it allows one to deal with problems involving non-analytic
functions of complex variables using real-variable methods. It has also been
used in \cite{pure} to obtain analytic and probabilistic results for a large
class of planar ensembles, building on previous work on time-frequency
analysis and approximation theory \cite{Abreusampling, AGR}. While
the papers \cite{pure,AGR} are concerned with universality-type
limit distribution laws and with probabilistic aspects of finite-dimensional
Weyl-Heisenberg ensembles, in this paper we will focus on \emph{infinite Weyl-Heisenberg ensembles}.

The article is organized as follows. Section \ref{sec_dpp} introduces planar DPP's and our model process, the Ginibre
ensemble. The notion of hyperuniformity is introduced in Section \ref{sec_hy}.
Section \ref{sec_poly} presents higher Landau levels and the polyanalytic ensembles, that describe them mathematically.
Section \ref{sec_wh} presents Weyl-Heisenberg ensembles and the main results on hyperuniformity, which are then
applied to the polyanalytic ensemble and higher Landau levels. We also discuss and analyze the total correlation
function and the structure factor associated with these point processes, and provide explicit formulae whenever
possible. Finally Section \ref{sec_con} summarizes the results.

\section{Planar DPP's and the Ginibre ensemble}
\label{sec_dpp}

Determinantal point processes are defined using an ambient space $\Lambda $
and a Radon measure $\mu $ defined on $\Lambda $ \cite{boro, DetPointRand}.
In our case, $\Lambda = \Rdst$, where $d=2\newd$ is an even number, and we use the identification $\Rdst=\bC^\newd$,
since several of the examples of interest are best described in terms of complex variables. For this reason, we
sometimes denote points in $\Rdst=\bC^\newd$ by $z$, instead of the usual $\mathrm{r}$.
The important object is the Hilbert space $L^{2}\left( \mathbb{C}^{\newd}\right) $,
with the Lebesgue measure in $\mathbb{C}^{\newd}$.
If $\{\varphi _{j}(z)\}_{j \geq 0}$ is an orthogonal sequence of $L^{2}\left(
\mathbb{C}^{\newd}\right) $, one can
define a reproducing kernel $K_{N}\left( z,w\right) $ by writing
\begin{equation}
K_{N}\left( z,w\right) =\sum_{j=0}^{N-1}\varphi _{j}(z)\overline{\varphi
_{j}(w)}.  \label{kernel}
\end{equation}
The kernel $K_{N}\left( z,w\right) $ will be the correlation kernel of the
point process $\mathcal{X}$, whose $k$-point intensities are given by $\rho
_{k}(x_{1},...,x_{k})=\det \left( K_{N}(x_{i},x_{j})\right) _{1\leq i,j\leq
k}$. For instance, if $\newd=1$, selecting $\varphi _{j}(z)=(\pi ^{j}/j!)^{\frac{
1}{2}}e^{-\tfrac{\pi}{2} \left\vert z\right\vert ^{2}}z^{j}$ for $j=0,...,N-1$ in (\ref
{kernel}), we obtain
\begin{equation*}
K_{N}\left( z,w\right) =e^{-\tfrac{\pi}{2} (\left\vert z\right\vert ^{2}+\left\vert
w\right\vert ^{2})}\sum_{j=0}^{N-1}\frac{(\pi z\overline{w})^{j}}{j!},
\end{equation*}
which is the correlation kernel of the \emph{Ginibre ensemble} of dimension $N$. If
we take $N\rightarrow \infty $, we obtain the correlation kernel of the
infinite Ginibre ensemble:
\begin{equation}
\label{eq_ginibre}
K_{\infty }\left( z,w\right) =e^{\pi z\overline{w}-\tfrac{\pi}{2} (\left\vert
z\right\vert ^{2}+\left\vert w\right\vert ^{2})}.
\end{equation}
The infinite Ginibre ensemble is \emph{translationally invariant}; this means that
the intensity functions satisfy:
$\rho_N(z_0+z,\ldots,z_{N-1}+z)=\rho_N(z_0,\ldots,z_{N-1})$, for all $z \in \bC$.

It is well known that the Ginibre ensemble is
equivalent to a model for the probability distribution of electrons in one
component plasmas \cite{Ginibre}. It also provides a model for the
statistical quantum dynamics of a charged particle evolving in a Euclidean
space under the action of a constant homogeneous magnetic field in the first
Landau level. The Ginibre ensemble can also be seen as a 2D electrostatic
model with $N$ unit charges interacting in a two dimensional space, which is
taken as the complex plane of the variable $z$. Indeed, if the potential
energy of the system is given as
\begin{equation*}
U\left( z_{0},..,z_{N-1}\right) =-\sum_{0\leq i<j\leq N-1}\log \left\vert
z_{i}-z_{j}\right\vert +\pi \sum_{k=0}^{N-1}\left\vert z_{k}\right\vert ^{2}
,
\end{equation*}
the corresponding probability distribution of the positions $z_{0},\ldots,z_{N-1}$
when the charges are in thermodynamical equilibrium, is proportional to the
measure
\begin{equation}
\exp \left[ -U\left( z_{0},..,z_{N-1}\right) \right] =\exp \left[ -\pi
\sum_{k=0}^{N-1}\left\vert z_{k}\right\vert ^{2}\right] \prod_{0\leq i<j\leq
N-1}\left\vert z_{i}-z_{j}\right\vert ^{2}.  \label{vandermond}
\end{equation}
It has been shown by Jean Ginibre \cite{Ginibre, Mehta} that the distribution
associated with the measure (\ref{vandermond}) is proportional to the one
obtained from the $N$-point intensities associated with the Ginibre
ensemble of dimension $N$. Thus, the Ginibre ensemble provides a model for the distribution
of charged-like particles in the first Landau level. Until recently, there
was no similar Ginibre type model for higher Landau levels, but this gap in
the statistical physics literature is being filled thanks to recent
work concerning the \emph{polyanalytic Ginibre ensembles} \cite{pure,HendHaimi,HaiHen2,Haimi,Shirai} - see also Section
\ref{sec_poly}.

\section{Hyperuniformity of point processes}
\label{sec_hy}

\subsection{The number mean of a DPP}

The number mean of a point process is the average number of points expected
to be found inside an observation window $D\subset \Rst^n$. One can
obtain the number mean by integrating the single particle probability
density, which is proportional to the probability density of finding a
particle at a certain point $\mathrm{r} \in D$. In the case of a DPP, it can be
obtained from the $1$-point intensity $\rho $, which is simply defined as
the diagonal of the correlation kernel of the process:
\begin{equation*}
\rho (\mathrm{r})=K\left( \mathrm{r},\mathrm{r}\right) =\sum_{i}\left\vert \varphi
_{i}(\mathrm{r})\right\vert ^{2}.
\end{equation*}
The expected number of points $\mathcal{X}(D)$ to be found in $D\subset
\mathbb{C}$ is then given as
\begin{equation*}
\mathbb{E}\left[ \mathcal{X}(D)\right] =\int_{D}\rho (\mathrm{r})d\mathrm{r}.
\end{equation*}
The $1$-point intensity $\rho $ is also called the single particle
probability density.

\subsection{Number variance and hyperuniformity}

In \cite{TorStil}, it has been discovered that a hyperuniform many-particle
system modeled by a point process $\mathcal{X}$ in a Euclidean space of
dimension $d$\ (not necessarily determinantal) is one in which the number
variance
\begin{equation*}
\sigma ^{2}(R)=\mathbb{E}\left[ \mathcal{X}(D_{R})^{2}\right] -\mathbb{E}
\left[ \mathcal{X}(D_{R})\right] ^{2},
\end{equation*}
where $D_{R}$ is a $d$-dimensional ball of radius $R$, satisfies
\begin{equation*}
\sigma ^{2}(R) = o(R^{d}).
\end{equation*}

\subsection{The total correlation function}
We will be mostly interested in a translationally invariant point process of intensity 1, i.e.,
the intensity functions satisfy $\rho_n(\mathrm{r}_1+\mathrm{r},\ldots, \mathrm{r}_n+\mathrm{r})=
\rho_n(\mathrm{r}_1,\ldots, \mathrm{r}_n)$, for all $\mathrm{r} \in \Rst^n$, and
$\rho_1 \equiv 1$. For such processes, the two-point intensity depends essentially on one variable,
and we may write:
\begin{align}
\label{eq_h}
\rho_2(\mathrm{r}_1,\mathrm{r}_2) = 1 + h(\mathrm{r}_2-\mathrm{r}_1),
\end{align}
where $h$ is known as the \emph{total correlation function}, and is related to the determinantal kernel by
\begin{align}
\label{eq_h2}
\abs{K(\mathrm{r}_1,\mathrm{r}_2)}^2 =  -h(\mathrm{r}_2-\mathrm{r}_1).
\end{align}
In statistical mechanics, it is also common to consider the \emph{structure factor} defined by
\begin{align*}
S(\mathrm{k}) = 1 + \hat{h}(\mathrm{k}),
\end{align*}
in the reciprocal space (Fourier) variable $\mathrm{k}$. (We normalize the Fourier transform
as: $\hat{f}(\mathrm{k}) = \int f(\mathrm{r}) e^{i \mathrm{r} \cdot \mathrm{k}} d\mathrm{r}$.)

\section{Polyanalytic Ginibre ensembles and higher Landau levels}
\label{sec_poly}

\subsection{The Landau levels}

Polyanalytic ensembles of the pure type model the random distribution of
charged-liked electrons in the so-called\emph{\ Landau levels}. Let us briefly describe this
relation (see \cite{HendHaimi,AoP,Zouhair} for more
details). The Hamiltonian operator describing the dynamics of a particle of
charge $e$ and mass $m_{\ast }$ on the Euclidean $xy$-plane, while
interacting with a perpendicular constant homogeneous magnetic field, is
given by the operator
\begin{equation}
H:=\frac{1}{2m_{\ast }}\left( i\hbar \nabla -\frac{e}{c}\mathbf{A}\right)
^{2},  \label{2.1.1}
\end{equation}
where $\hbar $ denotes Planck's constant, $c$ is the light velocity and $i$
the imaginary unit. Denote by $B>0$ the strength of the magnetic field and
select the symmetric gauge
\begin{equation*}
\mathbf{A=-}\frac{\mathbf{r}}{2}\times \mathbf{B=}\left( -\frac{B}{2}y,\frac{
B}{2}x\right) ,
\end{equation*}
where $\mathbf{r}=\left( x,y\right) \in \mathbb{R}^{2}$. For simplicity, we
set $m_{\ast }=e=c=\hbar =1$ in (\ref{2.1.1}), leading to the Landau
Hamiltonian
\begin{equation}
H_{B}:=\frac{1}{2}\left( \left( i\partial _{x}-\frac{B}{2}y\right)
^{2}+\left( i\partial _{y}+\frac{B}{2}x\right) ^{2}\right)  \label{LandauH}
\end{equation}
acting on the Hilbert space $L^{2}\left( \mathbb{R}^{2},dxdy\right) $. The
spectrum of the Hamiltonian $H_{B}$ consists of an infinite number of
eigenvalues with infinite multiplicity of the form
\begin{equation}
\epsilon _{n}^{B}=\left( n+\frac{1}{2}\right) B,\text{ \ \ \ \ }n=0,1,2,...
\label{2.1.4}
\end{equation}
Without loss of generality, we set $B=2\pi $ to simplify the relation to the
Weyl-Heisenberg group described in the next sections. Then we define the
operator $L_{z}$ by conjugating the Landau Hamiltonian (\ref{LandauH}) as
follows:
\begin{equation}
L_{z}:=e^{\tfrac{\pi}{2} \left\vert z\right\vert ^{2}}\left( \frac{1}{2}H_{2\pi }-\frac{\pi}{2}
\right) e^{-\tfrac{\pi}{2} \left\vert z\right\vert ^{2}}=-\partial _{z}\partial _{
\overline{z}}+\pi \overline{z}\partial _{\overline{z}},
\label{2.1.3}
\end{equation}
acting on the Hilbert space $L^{2}\left( \mathbb{C}\right) $. The spectrum
of $L_{z}$ is given by $\sigma (L_{z})=\{\nu \pi: \nu=0,1,2,\ldots\}$. The eigenvalue
$r \pi$ is the Landau level
of order $r$.\ The eigenspace associated with the eigenvalue $r \pi $ is called the \emph{pure Landau level eigenspace
of order}
$r$. With each pure Landau level eigenspace of order $r$ one can associate
correlation kernels of the form
\begin{equation}
K_{r}(z,w)=L_{r}^{0}(\pi \left\vert z-w\right\vert ^{2})e^{\pi z\overline{w}
-\tfrac{\pi}{2} (\left\vert z\right\vert ^{2}+\left\vert w\right\vert ^{2})},
\label{purekernel}
\end{equation}
where $L_{r}^{0}$ is the Laguerre polynomial, defined, for a general parameter
$\alpha $, as
\begin{equation*}
L_{n}^{\alpha }(x)=\sum_{k=0}^{n}(-1)^{k}\binom{n+\alpha }{n-k}\frac{x^{k}}{
k!}.
\end{equation*}
As we will see in the next section, (\ref{purekernel}) is the reproducing
kernel of \emph{a pure Fock space of polyanalytic functions}.\ Thus, we name
the resulting determinantal point process as a \emph{polyanalytic ensemble
of the pure type}. It is related to the polyanalytic Ginibre ensembles
investigated in \cite{HendHaimi} and the terminology \emph{pure type }is
inherited from the Landau level interpretation: determinantal processes with
kernels of the form considered in \cite{HendHaimi} have a physical
interpretation as probabilistic 2D models for the distribution of
electrons in the first $N$ Landau levels, while processes with correlation
kernels of the form (\ref{purekernel}) model the distribution of electrons
in a \emph{pure} Landau level of order $r$. For reference, one can keep in
mind that the basis functions of the Ginibre ensemble generate the proper
subspace of $L^{2}\left( \mathbb{C}\right) $ consisting of analytic
functions (the so-called Bargmann-Fock space). Moreover, the case $r=0$ in
\eqref{purekernel} is simply
\begin{equation*}
K_{0}(z,w)=e^{\pi z\overline{w}-\tfrac{\pi}{2} (\left\vert z\right\vert ^{2}+\left\vert
w\right\vert ^{2})},
\end{equation*}
which is the correlation kernel of the infinite Ginibre ensemble. Thus, the
polyanalytic ensemble of the pure type associated with the first Landau
level is, as mentioned in the introduction, the Ginibre ensemble.

Using the formula for the kernel in \eqref{purekernel} and \eqref{eq_h2}, we see that the \emph{total correlation
function of the polyanalytic ensemble of the pure type} is:
\begin{align}
\label{eq_hr}
h_r(z) = - \left[ L^0_r \left( \pi \abs{z}^2 \right) \right]^2 e^{-\pi \abs{z}^2},
\qquad z \in \bC.
\end{align}
We note that $h_r$ is a radial function and
that
$\abs{h_r(z)} \leq C_\alpha e^{-\alpha\pi \abs{z}^2}$,
for every $\alpha \in (0,1)$ - where $C_\alpha$
is a constant that depends on $\alpha$.

\subsection{Polyanalytic Fock spaces}

A function $F(z,\overline{z})$, defined on a subset of $\mathbb{C}$, and
satisfying the generalized Cauchy-Riemann equations
\begin{equation}
\left( \partial _{\overline{z}}\right) ^{q}F(z,\overline{z})=\frac{1}{2^{q}}
\left( \partial _{x}+i\partial _{\xi }\right) ^{q}F(x+i\xi ,x-i\xi )=0\text{,
}  \label{eq:c1}
\end{equation}
is said to be \emph{polyanalytic of order }$q-1$
\cite{balk}. It is clear from (\ref
{eq:c1}) that the following polynomial of order $q-1$ in $\overline{z}$\
\begin{equation}
F(z,\overline{z})=\sum_{k=0}^{q-1}\overline{z}^{k}\varphi _{k}(z),
\label{polypolynomial}
\end{equation}
where the coefficients $\{\varphi _{k}(z)\}_{k=0}^{q-1}$ are analytic
functions,\ is a polyanalytic function\ of order $q-1$. By solving $\partial
_{\overline{z}}F(z,\overline{z})=0$, an iteration argument shows that every $
F(z,\overline{z})$ satisfying (\ref{eq:c1}) is indeed of the form (\ref
{polypolynomial}).

The polyanalytic Fock space $\mathbf{F}^{q}(\mathbb{C})$
consists of all the functions of the form $e^{-\tfrac{\pi}{2}
\left\vert
z\right\vert ^{2}}F(z,\overline{z})$, with $F(z,\overline{z})$ polyanalytic
functions of order $q-1$, supplied with the Hilbert space structure of $
L^{2}(\mathbb{C})$. The\ (infinite-dimensional) kernel of the polyanalytic
Fock space $\mathbf{F}^{q}(\mathbb{C})$\ is
\begin{equation}
\mathbf{K}^{q}(z,w)=L_{q}^{1}(\pi \left\vert z-w\right\vert ^{2})e^{\pi z
\overline{w}-\tfrac{\pi}{2} (\left\vert z\right\vert ^{2}+\left\vert w\right\vert ^{2})}
.  \label{polyFock}
\end{equation}
The connection to the Landau levels - see \cite{Abreusampling} for
applications of this connection to signal analysis and \cite{HendHaimi,AoP} to physics - follows from the following
orthogonal decomposition,
first observed by Vasilevski \cite{VasiFock}:
\begin{equation}
\mathbf{F}^{q}(\mathbb{C})=\mathcal{F}^{0}(\mathbb{C})\oplus ...\oplus
\mathcal{F}^{q-1}(\mathbb{C}),
\label{orthogonal}
\end{equation}
where $\mathcal{F}^{r}(\mathbb{C})$ is the pure Landau level eigenspace of order $r$, whence the terminology \emph{pure}
used in \cite{HendHaimi}. \emph{Pure} \emph{poly-Fock spaces }provide a full
orthogonal decomposition of the whole $L^{2}(\mathbb{C})$:
\begin{equation*}
L^{2}(
\mathbb{C}
)=\bigoplus_{r=1}^{\infty }\mathcal{F}^{r}(
\mathbb{C}
).
\end{equation*}
The formula for Laguerre polynomials $\sum_{r=0}^{q-1}L_{r}^{\alpha
}=L_{q-1}^{\alpha +1}$\ and (\ref{orthogonal})\ show that
\begin{equation*}
\mathbf{K}^{q}(z,w)=\sum_{r=0}^{q-1}\mathcal{K}^{r}(z,w),
\end{equation*}
where $\mathcal{K}^{r}(z,w)$ is the reproducing kernel (\ref{purekernel}) of
the pure Landau level eigenspace of order $r$.

\section{The Weyl-Heisenberg ensembles}
\label{sec_wh}

In this section we work with functions of several real variables, but keep a multi-index notation similar to the
univariate case. As before, we let $d=2\newd$ be an even positive integer.

\subsection{The Schr\"{o}dinger representation of the Heisenberg group}
The infinite Weyl-Heisenberg ensembles are DPP's associated with the representation of the Heisenberg group
(in \cite{pure}, finite-dimensional versions are investigated). Given a
\emph{window function} $g\in L^{2}({\mathbb{R}^{\newd}})$, the Schr\"{o}dinger
representation of the Heisenberg group $\mathbb{H}$ acts on $L^{2}({
\mathbb{R}^{\newd}})$ by means of the unitary operators
\begin{equation*}
T(x,\xi ,\tau )g(t)=e^{2\pi i\tau }e^{-\pi ix\xi }e^{2\pi i\xi t}g(t-x),
\qquad (x, \xi) \in \Rdst, \tau \in \Rst.
\end{equation*}
The corresponding representation coefficients are
\begin{equation*}
\left\langle f,T(x,\xi ,\tau )g\right\rangle =e^{-2\pi i\tau }e^{\pi ix\xi
}\left\langle f,e^{2\pi i\xi \cdot}g(\cdot-x)\right\rangle .
\end{equation*}

\subsection{Time-frequency analysis}

The short-time Fourier transform $V_{g}f(x,\xi )$ can be defined in terms of
the above representation coefficients by eliminating the variable $\tau$ as follows:
\begin{equation*}
V_{g}f(x,\xi )=e^{2\pi i\tau }e^{-\pi ix\xi }\left\langle f,T(x,\xi ,\tau
)g\right\rangle =\left\langle f,e^{2\pi i\xi \cdot}g(\cdot-x)\right\rangle .
\end{equation*}
We introduce convenient notation where we identify a pair $(x,\xi )\in {
\mathbb{R}^{d}}$ with the complex vector $z=x+i\xi \in {\mathbb{C}^{\newd}}$.
The \emph{time-frequency shifts} of a function $g:{\mathbb{R}^{\newd}}
\rightarrow {\mathbb{C}}$ are defined as follows:
\begin{equation*}
\pi (z)g(t):=e^{2\pi i\xi
t}g(t-x),\qquad z=(x,\xi )\in {\mathbb{R}^{\newd}}\times {\mathbb{R}^{\newd}},
\quad t\in {\mathbb{R}^{\newd}}.
\end{equation*}
With this notation, given a window function $g\in L^{2}({\mathbb{R}^{\newd}})$,
the \emph{short-time Fourier transform} of a function $f\in L^{2}({\mathbb{R}
^{\newd}})$ with respect to $g$ is
\begin{equation*}
V_{g}f(z):=\left\langle f,\pi (z)g\right\rangle ,\qquad z\in {\mathbb{R}^{2\newd}
}.
\end{equation*}
The subspace of $L^{2}({\mathbb{R}^{2\newd}})$ which is the image of $L^{2}({
\mathbb{R}^{\newd}})$ under the short-time Fourier transform with the window $g$,
\begin{equation*}
\mathcal{V}_{g}=\left\{ V_{g}f:f\in L^{2}({\mathbb{R}^{\newd}})\right\} \subset
L^{2}({\mathbb{R}^{2\newd}}),
\end{equation*}
is a Hilbert space with reproducing kernel given by
\begin{equation}
K_{g}(z,w)=\left\langle \pi({w})g,\pi({z})g\right\rangle _{L^{2}({\mathbb{R}
^{\newd}})}.  \label{repdef}
\end{equation}
With the notation $z=(x,\xi), w=(x',\xi')$, the kernel can be written explicitly as
\begin{equation}
\label{eq_kg_expl}
K_{g}(z,w)= \int_{\Rst^\newd}
\overline{g(t-x)} g(t-x') e^{2 \pi i t(\xi'-\xi)} dt.
\end{equation}
We can now introduce the WH ensembles.
\begin{definition}
Let $g\in L^{2}(\mathbb{R}^{\newd})$ be of norm 1 and such that
\begin{equation}
\label{eq_decay}
\left\vert V_{g}g(z)\right\vert \leq C \left( 1 + \left\vert z \right\vert \right)^{-s}
<+\infty ,
\end{equation}
for some $s>2\newd+1$ and $C>0$.
\emph{The infinite Weyl-Heisenberg ensemble} associated with the
function $g\in L^{2}(\mathbb{R}^{\newd})$ is the determinantal point process
with correlation kernel
\begin{equation*}
K_{g}(z,w)=\left\langle \pi({w})g,\pi({z})g\right\rangle _{L^{2}({\mathbb{R}^{\newd}})}.
\end{equation*}
\end{definition}
\begin{remark}
{\normalfont
The condition in \eqref{eq_decay} amounts to decay of $g$ in both the space and frequency variables,
and is satisfied by any Schwartz-class function.
}
\end{remark}
\begin{remark}
{\normalfont
The WH ensemble associated with a window $g$ is well-defined due to the Macchi-Soshnikov theorem
\cite{ma,so}. Indeed, the kernel $K_{g}$ represents a projection operator and we only need to verify
that it is locally trace-class. Given a compact domain $D$, the operator $T_{g,D}$ represented by
the localized kernel $K_{g,D}$ is known as a Gabor-Toeplitz operator. It is well-known that
$T_{g,D}$ is trace-class and that $\trace(T_{g,D})=\abs{D}$; see for example \cite{Charly, AGR}.
}
\end{remark}
\begin{remark}
\label{rm_ani}
\normalfont{
For general windows $g$, the resulting WH ensemble is \emph{statistically anisotropic} in the sense that the point-statistics may depend on
the vector displacements
between the points.
}
\end{remark}

Figure \ref{fig} shows realizations of Weyl-Heisenberg ensembles corresponding to
two different windows: the Gaussian and the Hermite function of order 7. As explained
in Section \ref{sec_xx}, these correspond to different Landau levels.

While the correlation kernel of a WH ensemble is not translationally invariant, a simple calculation shows that the
corresponding point process is. In addition, using the explicit formula for the kernel in \eqref{eq_kg_expl} and
\eqref{eq_h2}, we see that \emph{the total correlation function of an infinite WH ensemble} is:
\begin{align}
\label{eq_hg}
h_g(z)=-\abs{\int_{\Rst^\newd} g(t-x) \overline{g(t)} e^{2 \pi i t \xi} dt}^2
=-\abs{V_g g(z)}^2,
\qquad z=(x,\xi) \in \Rst^\newd \times \Rst^\newd = \Rst^d.
\end{align}
As we show in Section \ref{sec_xx}, for concrete choices of the underlying window function $g$, it is possible to get
explicit expressions for the corresponding function $h_g$. Moreover, in many cases,
these are radial functions. We next show that, in that case, we can describe the asymptotics of the structure factor
near the origin.

\begin{lemma}
Assume that the correlation function $h_g$ of a WH ensemble is radial. Then the corresponding structure factor
satisfies:
\begin{align*}
\abs{S_g(\mathrm{k})} \asymp \abs{\mathrm{k}}^2, \qquad \mbox{ as } \mathrm{k} \longrightarrow 0.
\end{align*}
More precisely, there exist two constants $c,C>0$ such that $c\abs{\mathrm{k}}^2 \leq
\abs{S_g(\mathrm{k})} \leq C \abs{\mathrm{k}}^2$, for $\mathrm{k}$ near 0.
\end{lemma}
\begin{proof}
Using \eqref{eq_decay} and \eqref{eq_hg}, we see that $\abs{h_g(z)} \lesssim (1 + \abs{z})^{-(2 d +2)}$, and
consequently
\[\left( 1 + \abs{z}^2 \right) \abs{h_g(z)} \in L^1(\Rst^d, dz).\]
Since $S_g(\mathrm{k})=1+\hat{h}_g(\mathrm{k})$,
it follows
that the structure factor $S_g$ is a $C^2$ function, and we can Taylor expand it as
\begin{align}
\label{eq_taylor}
S_g(\mathrm{k}) = S_g(0) + \sum_{j=1}^d \partial_{\mathrm{k}_j} S_g(0) \mathrm{k}_j
+ \sum_{\abs{\alpha}=2} \partial^\alpha_{\mathrm{k}} S_g(0) \mathrm{k}^\alpha + o(\abs{\mathrm{k}}^2).
\end{align}
First, note that $S_g(0)=1+h_g(0)=1-\norm{g}_2^2=0$. Second, since $h_g$ is radial,
\begin{align*}
\partial_{\mathrm{k}_j} S_g(0) = -i \int_{\Rst^d} z_j h_g(z) dz = 0,
\end{align*}
and, therefore, the linear terms in \eqref{eq_taylor} vanish.
Similarly, the cross second derivatives in \eqref{eq_taylor} - $\partial_{\mathrm{k}_j} \partial_{\mathrm{k}_{j'}}
S(0)$, with $j\not=j'$ - vanish, leading to
\begin{align*}
S_g(\mathrm{k}) = \Delta_{\mathrm{k}} S_g(0) \abs{\mathrm{k}}^2 + o(\abs{\mathrm{k}}^2).
\end{align*}
Hence, it suffices to show that $\Delta_{\mathrm{k}} S_g(0) \not=0$. This is the case because
\begin{align*}
\Delta_{\mathrm{k}} S_g(0) = - \int_{\Rst^d} \abs{z}^2 h_g(z) dz,
\end{align*}
and $h_g \leq 0$, while $h_g \not\equiv 0$.
\end{proof}

\begin{figure}
\centering
\subfigure[Landau level \# 1.]{
\includegraphics[scale=0.5]{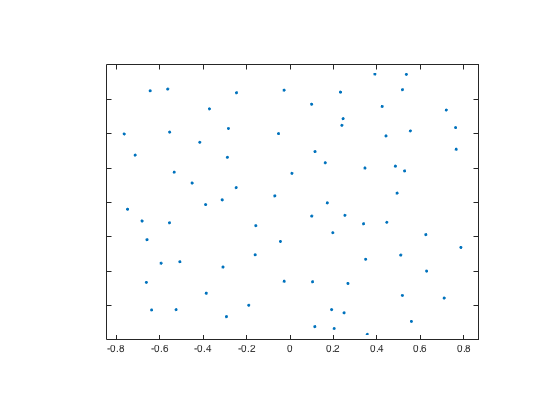}
\label{fig_levels}
}
\subfigure[Landau level \# 7.]
{
\includegraphics[scale=0.5]{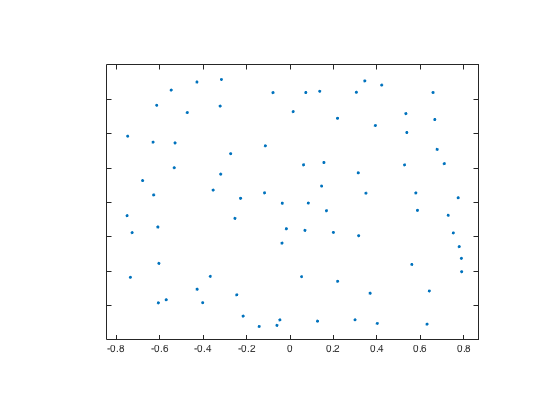}
}
\caption{WH ensembles corresponding to different Hermite windows.}
\label{fig}
\end{figure}

\subsection{Hyperuniformity of infinite Weyl-Heisenberg ensembles}

Now we will present our study of the variance of Weyl-Heisenberg
ensembles, relying on spectral methods originating from time-frequency
analysis. Given a set $\Omega \subseteq \mathbb{R}^{d}$, $\partial
\Omega $ denotes its frontier, $1_{\Omega }$ its characteristic
function, $\left\vert \Omega \right\vert $ its measure and $\left\vert \partial
{\Omega }\right\vert $ its perimeter (defined as the $d-1$ dimensional
measure of the boundary). The following is our main result.

\begin{theorem}
\label{th_1}
Weyl-Heisenberg ensembles are hyperuniform. More precisely, if
$\mathcal{X}$ is a WH-ensemble and $D_{R}\subset
\mathbb{R}^{d}$ is a $d$-dimensional ball of radius $R$, then, as $R\rightarrow \infty $,
\begin{equation*}
\sigma ^{2}(R)=\mathbb{V}\left[ \mathcal{X}(D_{R})\right] \lesssim R^{d-1}.
\end{equation*}
\end{theorem}

\begin{proof}
We want to show that the number variance $\sigma ^{2}(R)=\mathbb{E}\left[
\mathcal{X}(D_{R})^{2}\right] -\mathbb{E}\left[ \mathcal{X}(D_{R})\right]
^{2}$ - where $D_{R}$ is a $2\newd$-dimensional ball of radius $R$ - satisfies $
\sigma ^{2}(R)\lesssim R^{2\newd-1}$. We start defining the concentration operator
\begin{equation*}
(T_{D_{R}}f)(z)=\int_{D_{R}}f(w) K_{g}(z,w) dw.
\end{equation*}
Using (\cite[Equation (1.2.4)]{DetPointRand}) one can write the number variance of $
\mathcal{X}(D_{R})$\ as
\begin{eqnarray*}
\sigma ^{2}(R) &=&\mathbb{E}\left[ \mathcal{X}(D_{R})^{2}\right] -\mathbb{E}
\left[ \mathcal{X}(D_{R})\right] ^{2} \\
&=&\int_{D_{R}}K_{g}\left( z,z\right) dz-\int_{D_{R}\times D_{R}}\left\vert
K_{g}\left( z,w\right) \right\vert ^{2}dzdw \\
&=&\trace\left( T_{D_{R}}\right) -\trace\left( T_{D_{R}}^{2}\right) \\
&=&\left\vert D_{R}\right\vert -\int_{D_{R}\times D_{R}}\left\vert
K_{g}(z,w)\right\vert^2 dzdw.
\end{eqnarray*}
Now, one can use \cite[Proposition 3.4]{AGR} to obtain the upper inequality
for the number variance:
\begin{equation}
\sigma ^{2}(R)\lesssim \left\vert \partial D_{R}\right\vert \int_{\mathbb{R}
^{2\newd}}\left\vert z\right\vert \left\vert V_{g}g(z)\right\vert ^{2}dz\lesssim
R^{2\newd-1}.  \label{varsup}
\end{equation}
(See \cite{Landau1967} for applications of this kind of inequalities to sampling theory.)
\end{proof}
\begin{remark}
\label{rem_more}
{\normalfont
The proof of Theorem \ref{th_1} extends to more general observation windows. In this case,
the number variance is dominated by the perimeter of the observation window.
}
\end{remark}

The next result shows that $R^{d-1}$ is actually the precise rate of
convergence.

\begin{theorem}
\label{th_var}
The variance of a Weyl-Heisenberg ensemble satisfies, as $R\rightarrow \infty $,
\begin{equation*}
\sigma ^{2}(R)=\mathbb{V}\left[ \mathcal{X}(D_{R})\right] \sim R^{d-1}\text{,
}
\end{equation*}
where $D_{R}$ is a $d$-dimensional ball of radius $R$.
\end{theorem}

\begin{proof}
As we have seen in the proof of Theorem \ref{th_1},
\begin{equation*}
\sigma ^{2}(R)=\trace\left( T_{D_{R}}\right) -\trace\left(
T_{D_{R}}^{2}\right) .
\end{equation*}
Arguing as in the proof of Lemma 3.3 in \cite{AGR}, we obtain the
following formula, where the variance is bounded in terms of the counting
function of the eigenvalues of $T_{D_{R}}$ - $\{\lambda_k(R): k\geq 1\}$ - that are
above a certain threshold. More precisely, for $\delta \in (0,1)$:
\begin{equation}
\sigma ^{2}(R)=\trace\left( T_{D_{R}}\right) -\trace\left(
T_{D_{R}}^{2}\right) \geq \frac{\left\vert \#\{k\geq 1:\lambda
_{k}(R)>1-\delta \}-\left\vert D_{R}\right\vert \right\vert }{\max \left\{
\frac{1}{\delta },\frac{1}{1-\delta }\right\} }.
\label{estimateEigenVariance}
\end{equation}
Now, by \cite[Theorem 4.1]{DeMarie}, there exists $\delta $ independent of
$R$ such that
\begin{equation*}
\left\vert \#\{k\geq 1:\lambda _{k}(R)>1-\delta \}-\left\vert
D_{R}\right\vert \right\vert \gtrsim R^{2\newd-1}.
\end{equation*}
Combining this with (\ref{estimateEigenVariance}) leads to the lower
inequality
\begin{equation*}
\sigma ^{2}(R)\gtrsim R^{2\newd-1},
\end{equation*}
which, together with the upper inequality \eqref{varsup}, yields the result.
\end{proof}

\begin{remark}
\normalfont{
Theorem \ref{th_var} implies that the central limit theorem of Costin and Lebowitz \cite{cl} (in the general formulation of Soshnikov \cite{so2}) is applicable and, therefore, the random variables
$\mathcal{X}(D_R)$ - when properly rescaled - are asymptotically normal as $R \rightarrow \infty$.
}
\end{remark}

\begin{remark}
\normalfont{
Theorem \ref{th_var} extends a result of Shirai \cite{Shirai}, that concerns DPP's that are translationally and rotationally invariant
(with a suitably decaying correlation kernel). In this case, asymptotic formulas for the implied constants are also available.
For general windows $g$, WH ensembles do not need to be rotationally invariant, cf.
Remark \ref{rm_ani}. It is noteworthy that the hyperuniformity concept has recently been generalized
to incorporate anisotropic features \cite{dir}
and thus the WH
ensembles provide a rigorous testbed
to study directional hyperuniformity.
}
\end{remark}

\subsection{Weyl-Heisenberg ensembles for higher Landau levels: polyanalytic
Ginibre-type ensembles}
\label{sec_xx}

For $m=1$, using the notation $z=x+i\xi $ and $w=u+i\eta $, a calculation (see \cite
{AbreuMM}) shows that the reproducing kernel of $\mathcal{V}_{h_{r}}$ is
related as follows to the reproducing kernel of the pure Fock space of
polyanalytic functions:
\begin{equation}
K_{h_{r}}(\overline{z},\overline{w})=e^{-i\pi (u\eta -x\xi )-\pi \frac{\left\vert z\right\vert
^{2}+\left\vert w\right\vert ^{2}}{2}}L_{r}^{0}(\pi \left\vert
z-w\right\vert ^{2})e^{\pi z\overline{w}}.  \label{relkern}
\end{equation}
Thus, the operator $E$ which maps
$F$ to
\begin{equation*}
e^{-i\pi x\xi }F(\overline{z})
\end{equation*}
is an isometric isomorphism
\begin{equation*}
E:\mathcal{V}_{h_{r}}\rightarrow \mathcal{F}^{r}(\mathbb{C}).
\end{equation*}

Thus, all properties of Weyl-Heisenberg ensembles are automatically
translated to the polyanalytic ensembles, in particular the hyperuniformity
property. In addition, polyanalytic ensembles, as presented in Section \ref{sec_poly}, extend
verbatim to $\bC^\newd = \Rst^d$, provided that the formulae are interpreted in a vectorial sense.
With this understanding, we obtain from Theorem \ref{th_1} the following corollary.

\begin{coro}
The pure polyanalytic ensembles are hyperuniform, and, as $R\rightarrow \infty $,
\begin{equation*}
\sigma ^{2}(R)=\mathbb{V}\left[ \mathcal{X}(D_{R})\right] \sim R^{d-1},
\end{equation*}
where $D_{R}$ is a $d$-dimensional ball of radius $R$.
\end{coro}

In the case $d=2$, this has been proved in \cite[Theorem 1]{Shirai} using explicit computations which
also provide the value of the asymptotic constant for the pure polyanalytic
Ginibre ensemble of order $r$, $C_{r}=\frac{8}{\pi ^{2}}r^{1/2}$. As noted in Remark \ref{rem_more},
the variance bounds in Theorem \ref{th_1} also apply to more general observation windows, and these conclusions
therefore extend to the polyanalytic ensembles.

\section{Conclusions}
\label{sec_con}
We introduced the infinite Weyl-Heisenberg ensembles
in $\Rdst$
and showed that they are hyperuniform.
This provides another class
of examples of $d$-dimensional
determinantal point processes
that are hyperuniform beyond
the so-called Fermi-type varieties \cite{TorScaZac}. We also proved that the number variance associated with spherical observation windows of radius $R$ grows like the surface
area of the window, $R^{d-1}$. Due to the Costin-Lebowitz central limit theorem, this implies
that the number of particles of a WH ensemble within a growing observation window are asymptotically
normal random variables.

We gave
explicit formulas for the total correlation functions of WH ensembles.
 In the radial case, we also derived asymptotics
near the origin for the structure factor, and showed that
$S(\mathrm{k}) \asymp \mathrm{k}^2$
in the limit $\mathrm{k} \rightarrow 0$ in all
space dimensions.
The two-dimensional
point process associated with the
Ginibre ensemble has similar
quadratic in $k$ structure-factor asymptotics.

Special choices of the waveform $g$ in the definition of a WH ensemble lead to important point processes. Specifically,
we showed that when $g$ is chosen as a Hermite function, the corresponding
point process coincides with the so-called polyanalytic Ginibre ensemble of the pure type, which
models the distribution of electrons in higher Landau levels. The corresponding total correlation
functions resemble the ones of
the Ginibre ensemble, this time with a Laguerre polynomial as an additional multiplicative factor. In particular, they decay, for large distances, faster than exponential; specifically, like a Gaussian.

The family of Weyl-Heisenberg ensembles also includes processes that are \emph{structurally
anisotropic} in the sense that the point-statistics depend on the different spatial directions.
Thus, our work provides the first rigorous means to study \emph{directional hyperuniformity} of
point processes. In such instances,
it is relevant to consider a more general notion of hyperuniformity that accounts
for the dependence of the structure factor on the direction in which the origin in Fourier space is
approached \cite{dir}.

\ack
Part of the research for this article was conducted while L. D. A. and J. L. R. visited
the Program in Applied and Computational Mathematics at Princeton
University. They thank PACM and in particular Prof. Amit Singer
for their kind hospitality. L. D. A. was supported by the Austrian Science Fund (FWF): START-project FLAME (”Frames
and Linear Operators for Acoustical Modeling and Parameter Estimation”, Y 551-N13).
J. M. P. was partially supported by AFOSR awards FA9550-12-1-0317 and FA9550-13-1-0076
(of his advisor Amit Singer). J. L. R. gratefully acknowledges support from the Austrian Science Fund (FWF):
P 29462 - N35.

\end{document}